\documentclass{IEEEtran}
\usepackage{amssymb,amsmath, amsthm, cite}
\usepackage{graphicx}
\usepackage{subfigure}
\usepackage{enumerate}
\usepackage{bbm}
\usepackage[printonlyused,withpage]{acronym}
\usepackage{pstool}

\newcommand{\s}{\mathrm{s}}
\newcommand{\PP}{\mathbb{P}}
\newcommand{\EE}{\mathbb{E}}
\newcommand{\RR}{\mathbb{R}}
\newcommand{\Ps}{\mathbb{P}_{\mathrm{s}}}
\newcommand{\Pout}{\mathbb{P}_{\mathrm{out}}}

\newcommand{\Pscan}{\mathbb{P}_{\mathrm{s,can}}}
\newcommand{\Psic}{\mathbb{P}_{\mathrm{s, IC}}}

\newcommand{\lm}{\lambda_{\mathrm{m}}}
\newcommand{\lk}{\lambda_{\mathrm{k}}}

\newcommand{\leqi}{\lambda_{\mathrm{eq}}}
\newcommand{\etat}{\eta_{\mathrm{t}}}

\newcommand{\muj}{\mu_j}

\newcommand{\pak}{p_{\mathrm{a},k}}
\newcommand{\pare}{p_{\mathrm{a}, k}^{\mathrm{(RE)}}}
\newcommand{\borRE}{C_k^{\mathrm{(RE)}}}
\newcommand{\Ri}{R_{\mathrm{I},n}}

\newcommand{\Prc}{\mathbb{P}_{\mathrm{c}}}
\newcommand{\PrcMax}{\mathbb{P}_{\mathrm{c}}^{(\text{MAX SIR})}}
\newcommand{\PrcMin}{\mathbb{P}_{\mathrm{c}}^{(\mathrm{MINL})}}
\newcommand{\be}{\begin{equation}}
\newcommand{\ee}{\end{equation}}
\newcommand{\bea}{\begin{eqnarray}}
\newcommand{\eea}{\end{eqnarray}}

\newtheorem{theorem}{Theorem}
\newtheorem{lemma}{Lemma}

\newtheorem{remark}{Remark}

\newacro{JD}{joint detection}
\newacro{IC}{interference cancellation}
\newacro{SIC}{successive interference cancellation}
\newacro{SIR}{signal-to-interference ratio}
\newacro{SoI}{signal of interest}
\newacro{MAC}{medium access control}
\newacro{DL}{downlink}
\newacro{UL}{uplink}
\newacro{PMF}{probability mass function}

\def\BibTeX{{\rm B\kern-.05em{\sc i\kern-.025em b}\kern-.08em
    T\kern-.1667em\lower.7ex\hbox{E}\kern-.125emX}}

\graphicspath{{./figure/}}

\begin{document}
\title{Successive Interference Cancellation in Heterogeneous Cellular Networks}

%\markboth{Submitted to IEEE Transactions on Signal Processing}{Wildemeersch \textit{\MakeLowercase{et al.}}: Successive Interference Cancellation in Heterogeneous Cellular Networks}

\author{Matthias Wildemeersch,~\IEEEmembership{Student Member,~IEEE,} Tony Q.~S.~Quek,~\IEEEmembership{Senior Member,~IEEE,} Marios Kountouris,~\IEEEmembership{Member,~IEEE,} Alberto Rabbachin,~\IEEEmembership{Member,~IEEE}  and Cornelis H. Slump, ~\IEEEmembership{Member,~IEEE} \\
\thanks{The material in this paper has been presented in part at the IEEE International Workshop on Signal Processing Advances in Wireless Communications (SPAWC), Darmstadt, Germany, June 2013 \cite{WilQueKouSlu_SPAWC13}.}
\thanks{M. Wildemeersch is with the Signals and Systems Group, University of Twente, the Netherlands, and the Institute for Infocomm Research, A$^{\ast}$STAR, Singapore (e-mail: \tt{stuwma@i2r.a-star.edu.sg}).}
\thanks{T.~Q.~S. Quek is with the Singapore University of Technology and Design, and the Institute for Infocomm Research (e-mail: \tt{qsquek@ieee.org}).}
\thanks{M. Kountouris is with the Department of Telecommunications, SUPELEC (Ecole Sup\'{e}rieure d'Electricit\'{e}), Gif-sur-Yvette, France (e-mail: \tt{marios.kountouris@supelec.fr}).}
\thanks{A. Rabbachin is with the Laboratory for Information
and Decision Systems (LIDS), Massachusetts Institute of Technology, Cambridge, USA (e-mail: \tt{rabalb@mit.edu}).}
\thanks{C. H. Slump is with the Signals and Systems Group, University of Twente, Enschede, the Netherlands (e-mail: \tt{C.H.Slump@utwente.nl}).}
}

\maketitle

%%%%%%%%%%%%%%%%%%%%%%%%%%%%%%%%%%%%%%
\vspace{-15mm}
\begin{abstract}
At present, operators address the explosive growth of mobile data demand by densification of the cellular network so as to reduce the transmitter-receiver distance and to achieve higher spectral efficiency. Due to such network densification and the intense proliferation of wireless devices, modern wireless networks are interference-limited, which motivates the use of interference mitigation and coordination techniques. In this work, we develop a statistical framework to evaluate the performance of multi-tier heterogeneous networks with  \ac{SIC} capabilities, accounting for the computational complexity of the cancellation scheme and relevant network related parameters such as random location of the access points (APs) and mobile users, and the characteristics of the wireless propagation channel. We explicitly model the consecutive events of canceling interferers and we derive the success probability to cancel the $n$-th strongest signal and to decode the signal of interest after $n$ cancellations. When users are connected to the AP which provides the maximum average received signal power, the analysis indicates that the performance gains of SIC diminish quickly with $n$ and the benefits are modest for realistic values of the \ac{SIR}. We extend the statistical model to include several association policies where distinct gains of SIC are expected: (i) minimum load association, (ii) maximum instantaneous SIR association, and (iii) range expansion. Numerical results show the effectiveness of SIC for the considered association policies. This work deepens the understanding of SIC by defining the achievable gains for different association policies in multi-tier heterogeneous networks.
%\mynote{Context}
%\mynote{What do we do} 
%\mynote{What do we add}
%\mynote{What can we do with it} 
\end{abstract}

%%%%%%%%%%%%%%%%%%%%%%%%%%%%%%%%%%%%%%
\begin{keywords}
Successive interference cancellation, multi-tier heterogeneous network, stochastic geometry, association policy
\end{keywords}

%\clearpage

%%%%%%%%%%%%%%%%%%%%%%%%%%%%%%%%%%%%%%
\section{Introduction}
Small cell networks are an important trend in current wireless networks, which increase the density of transmitters and result in interference-limited networks where the thermal noise is negligible with respect to the interference \cite{QueRocGuvKou_Cam13, WilQueRabSlu_TC2013}. The motivation of small cell networks stems from the idea to  reduce the distance between transmitter and receiver by deploying additional base stations as to increase the spectral efficiency. Yet, as the network interference is the most important obstacle for successful communication, effective interference management schemes are essential to further enhance the performance of dense networks. These mechanisms impose the orthogonality between transmitted signals in frequency, time, or space, and include adaptive spectrum allocation policies \cite{CheQueKou_JSAC2012, LopGuvRocKouQueZha_WComMag11}, \ac{MAC} schemes, spatial interference mitigation by means of zero-forcing beamforming \cite{KouAnd_TWC2012}, and signal processing algorithms usually referred to as \ac{IC} techniques \cite{verdu_1998, WebAndYanVec_TIT2007, WebAnd_NOW2012, HuaLauChe_JSAC2009, MorLoy_JSAC2009, HasAloBasEbb_TWC2003, BloJin_ICC2009, NguJeoQueTayShi_TWC13}.  

Signal processing techniques such as \ac{JD} \cite{verdu_1998} or \ac{SIC} reduce the interference power by decoding and canceling interfering signals. In this work, we focus on the SIC receiver which decodes signals according to descending signal power and subtracts the decoded signal from the received multi-user signal, so as to improve the \acf{SIR}. The process is repeated until the \acf{SoI} is decoded. A common approach in literature is to consider an exclusion region around the receiver. In \cite{Ina_TSP2012}, a Gaussian approximation is proposed for the sum of interfering signals, while \cite{WenLoyYon_JSAC12} performs an asymptotic analysis of the interference distribution in cognitive radio networks. The methodology based on the exclusion region leads to the definition of lower and upper bounds of the outage probability. For instance, in \cite{WebAndYanVec_TIT2007, WebAnd_NOW2012}, a stochastic geometric model is adopted to capture the spatial distribution of the interfering nodes, accounting for cancellation and decoding errors. The key idea of this work is the division into near field and far field interferers, where every near field interferer is able to cause outage at the reference receiver.  Building on this work, \cite{HuaLauChe_JSAC2009, MorLoy_JSAC2009} propose bounds of the outage/success probability including the effects of the fading channel, while \cite{BloJin_ICC2009} includes accurately the consecutive steps of the SIC scheme. None of these works concerns a specific cancellation technique, since the order statistics of the received signal power are disregarded, which is an essential aspect in the analysis of SIC. The ordering of the received signal power depends on the transmission power of the network nodes, the spatial distribution of the active transmitters, and the propagation channel conditions. Specifically, the inclusion of SIC in the network performance analysis requires the characterization of the fundamental information-theoretic metric, the SIR, and to model the network interference as a trimmed sum of order statistics. In \cite{HasAloBasEbb_TWC2003}, closed-form expressions are presented for the outage probability accounting for the order statistics, assuming that all interferers are at the same distance from the intended receiver, while \cite{BenCarGag_TC2010} derives a lower bound of the outage probability based on the order statistics of the strongest uncanceled and partially canceled signals accounting for distance and fading. Very recently, some excellent contributions can be found that explicitly include both the topology and the fading effects in the description of sum of order statistics of the received signal power \cite{ZhaHae_Globecom2012, ZhaHae_Asilomar13, ZhaHae_ISIT13, KeeBlaKar_ISIT13}, yet, these models limit the analysis to single-tier networks.

A unified approach to describe the performance of \ac{SIC}, which jointly accounts for the interference cancellation scheme, network topology, channel fading, and the specific aspects of multi-tier networks, is still evasive. In this work, we develop an analytical framework that describes the success probability for transmissions in multi-tier networks with \ac{SIC} capabilities. The main contributions of this work are listed as follows.
\begin{itemize}
\item We derive the probability of successfully canceling the $n$-th interferer and show that the order statistics of the received signal power are dominated by path loss attenuation. We demonstrate how the effectiveness of the SIC scheme depends on the path loss exponent, the density of users and APs, and the maximum number of cancellations. 
\item We provide an analytic framework that accommodates for the heterogeneity that characterizes future wireless networks. To this end, we include different association policies for multi-tier networks, for which SIC yields distinct performance gains. In particular, we include the minimum load association policy, maximum instantaneous SIR policy, and range expansion in the analysis.
\end{itemize}
The proposed framework accounts for all essential network parameters and provides insight in the achievable gains of SIC in multi-tier heterogeneous networks. The minimum load policy can be used to enhance the feasibility of load balancing \cite{YeRonCheAlsCarAnd_TWC13}, the maximum instantaneous SIR policy is relevant to define the achievable capacity in multi-tier networks \cite{DhiGanBacAnd_JSAC12}, and range expansion with SIC capabilities can allow for efficient traffic offloading \cite{SinDhiAnd_TWC13}.

The remainder of the paper is organized as follows. In Section \ref{sec.SysMod}, the system model is introduced. In Section \ref{sec.SIC}, the success probability of transmissions in multi-tier heterogeneous networks with SIC capabilities is defined. In Section \ref{sec.SICBen}, several association policies different from the maximum long-term SINR policy are introduced, where SIC yields a substantial increase of the coverage probability or the rate distribution. Numerical results are presented in Section \ref{sec.NumRes} and conclusions are drawn in Section \ref{sec.Con}.

%%%%%%%%%%%%%%%%%%%%%%%%%%%%%%%%%%%%%%
\section{System Model} \label{sec.SysMod}
We consider a multi-tier heterogeneous network composed of $K$ tiers. For every tier $k \in \mathcal{K}=\{1,...,K \}$, the access points (APs) are distributed according to a homogeneous Poisson point process (PPP) $\Phi_k$ in the Euclidean plane with density $\lk$ such that $\Phi_k \sim \text{PPP}(\lk)$. While it is natural to use the Poisson model as the underlying spatial stochastic process for irregularly deployed APs such as picocells and femtocells, modeling the location of regularly deployed macrocell base stations (MBSs) by means of a PPP has been empirically validated and yields conservative bounds on the network performance \cite{AndBacGan_TC2011}. More recently, also theoretical evidence has been given for modeling the deterministic locations of MBSs by means of a PPP, provided there is sufficiently strong log-normal shadowing \cite{BlaKarKee_INFOCOM13}. All APs apply an open access (OA) policy, such that users can be served by each AP of each tier. The mobile users are spatially distributed as $\Psi \sim \text{PPP}(\mu)$ over $\mathbb{R}^2$. Each AP of tier $k$ transmits with power $P_k$ over the total bandwidth $W$. The total available spectrum $W$ is divided in subchannels by aggregating a fixed number of consecutive subcarriers of bandwidth $B$, such that the total number of available subchannels equals $\lfloor W/B \rfloor$.\footnote{Without loss of generality, we assume $B = 1$.} We denote the subchannel index as $j$, where $j \in \mathcal{J} = \{ 1,2,\hdots,\lfloor W/B \rfloor\}$. In order to maximize frequency reuse and throughput, each AP has access to the entire available spectrum. We represent the $i$-th AP of tier $k$ as $x_{k,i}$. Hence, denoting the available channels of $x_{k,i}$ as $\mathcal{J}^{(x_{k,i})}$, we have $\mathcal{J}^{(x_{k,i})} = \mathcal{J}, \enspace \forall i, k$. A user receives a signal from $x_{k,i}$ with signal power $P_k h_u g_{\alpha}(u - x_{k,i})$, where $h_u$ represents the power fading coefficient for the link between the user $u$ and $x_{k,i}$, and $g_{\alpha}(x) =\| x\|^{-\alpha} $ is the power path loss function with path loss exponent $\alpha$. For notational convenience, $u$ and $x$ will be used to denote network nodes as well as their location. The association of a user to $x_{k,i}$ is based on the following association metric
\be \label{eq.assMax}
x_{k,i} = \arg \, \max_{k,i}\, A_k \| u-x_{k,i} \|^{-\alpha}, 
\ee
where $A_k$ represents the association rule. For all $A_k = 1$, the user is associated to the nearest base station. For $A_k = P_k$, the association is based on the maximum average received signal power, where the averaging is done with respect to the fading parameter $h$.\footnote{The association rule can further be adjusted to accommodate for cell range expansion by defining $A_k = b_k P_k$, where $b_k$ represents an association bias for tier $k$.} Using this association rule, the set of APs forms a multiplicatively weighted Voronoi tessellation on the two dimensional plane, where each cell $C_{k,i}$ consists of those points which have a higher average received signal power from $x_{k,i}$ than from any other AP, as depicted in Fig. \ref{fig.Voronoi}. 
\begin{figure}[t!]
\centering
\includegraphics[width = 1 \linewidth] {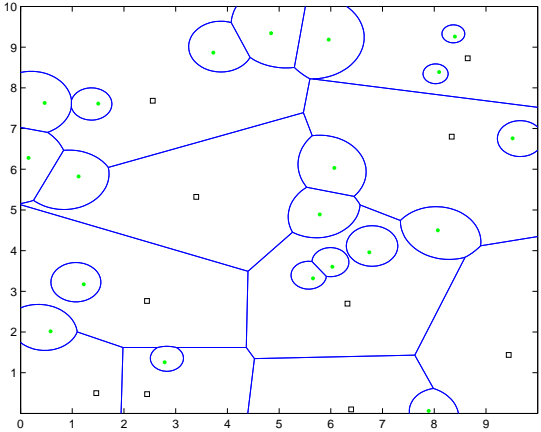}
\caption{Multiplicatively weighted Voronoi tessellation for a two-tier network. } 
\label{fig.Voronoi}
\end{figure}
Formally, we define the cells as
\begin{multline}
C_{k,i} = \lbrace y \in \RR^2 \enspace|\enspace \|y-x_{k,i}\| \leq (A_k/A_l)^{1/\alpha}\|y-x\|, \enspace \\
\forall \,x \in  \Phi_l \backslash\{x_{k,i}\}, \, l \in \mathcal{K}\rbrace .
\end{multline}
According to the association rule in \eqref{eq.assMax}, users will connect to different tiers and the density of users connected to tier $k$ is given by $\mu_k$. Considering a $K$-tier network, each tier $k \in \mathcal{K}$ is characterized by the set consisting of the \ac{UL} transmission power, \ac{DL} transmission power, AP density, and associated user density $\lbrace Q_k, P_k, \lambda_k, \mu_k \rbrace$. The sets of transmission powers and densities are denoted as $\mathbf{Q} = \lbrace Q_1, \hdots, Q_K \rbrace$, $\mathbf{P}= \lbrace P_1, \hdots, P_K \rbrace$, $\boldsymbol{\lambda}= \lbrace \lambda_1, \hdots, \lambda_K \rbrace$, and $\boldsymbol{\mu} = \lbrace \mu_1, \hdots, \mu_K \rbrace$ , respectively. Within a Voronoi cell, mobile users are independently and uniformly distributed over the cell area. Fairness between users is accomplished by proportional allocation of the time and frequency resources. We consider an orthogonal multiple access scheme, which ensures that at any given time and channel, only a single user per cell is active. As the number of users per Voronoi cell active on channel $j$ is restricted to one, the multiple access scheme introduces coupling between the locations of mobile users and APs. It can be shown that this dependence has negligible effects on the performance analysis, and in the sequel we will therefore assume independent PPPs to maintain the tractability of the system model \cite{DhiNovAnd_Globecom2012}. As interference dominates noise in modern cellular networks, we consider the network to be interference-limited. For the link between user $u$ and base station $x_{k,i}$, we define the signal-to-interference ratio (SIR) on channel $j$ in \ac{DL} as
\be
\text{SIR}_j(x_{k,i} \to u) = \frac{P_k h_u g_{\alpha}(x_{k,i}-u)}{\sum_{k \in \mathcal{K}} \sum_{v \in \Phi_{k,j}\backslash \{x_{k,i}\}} P_k h_v g_{\alpha}(v-u)} .
\ee
Let $\Phi_{k,j}$ and $\Psi_{k,j}$ denote the network nodes active on channel $j$ in tier $k$ for the APs and mobile users, respectively. A transmission is successful if the SIR of the intended link exceeds a prescribed threshold $\etat$, which reflects the required quality-of-service (QoS) in terms of transmission rate. Hence, the success probability can be written as $\PP_\s(\etat) = \Pr\{\text{SIR}_j(x_{k,i} \to u) \geq \etat\}$. 

%%%%%%%%%%%%%%%%%%%%%%%%%%%%%%%%%%%%%%
\section{Successive Interference Cancellation} \label{sec.SIC}
\newcounter{MYtempeqncnt}
\begin{figure*}[!t]
\normalsize
\setcounter{MYtempeqncnt}{\value{equation}}
\setcounter{equation}{8}
\be \label{eq.Psic}
\Psic (\etat, n) \approx  \int_{\Ri}^{\infty}\exp \left( -\pi  \muj \etat^{2/\alpha} u^2 C(\Ri^2/(\etat^{2/\alpha} u^2), \alpha)\right)  2\pi \leqi u \exp(-\leqi \pi u^2) \mathrm{d}u, 
\ee
\hrulefill
\setcounter{equation}{\value{MYtempeqncnt}}
\end{figure*}
In this section, we study how \ac{SIC} affects the success probability in multi-tier heterogeneous networks where the association policy is based on the maximum average received signal power. The analysis will be presented for \ac{UL} transmissions but the results similarly apply for \ac{DL} transmissions. This choice is motivated by the higher computational capabilities of APs in comparison with the mobile nodes, which is essential to successfully implement advanced signal processing techniques, as well as power consumption considerations which are less restrictive for APs. The concept of SIC is to decode the strongest signal and subtract it from the incoming signal which yields an increase of the SIR. In the analysis, we explicitly model the sequence of events in the cancellation process. We define the success probability as a function of the threshold, the number of canceled interferers, and all relevant system parameters such as the interferer density, transmission power, path loss exponent, and channel fading. 

Owing to constraints on computational complexity and delay, the number of interferers that can be canceled is limited to $N \in \mathbb{N}$. The AP with SIC capabilities attempts first to decode the \ac{SoI} without any interference cancellation. If an outage occurs, the AP seeks to decode the strongest signal, subtract it from the incoming signal, and performs a new attempt to decode the \ac{SoI} \cite{BloJin_ICC2009}. The received signal power at the typical AP can be ordered as $\{X_{(1)}, X_{(2)}, ...\}$ such that $X_{(i)} \geq X_{(j)}$, with $i \leq j$ and $X_{(i)} = Q_i h_i v_i^{-\alpha}$. The same actions are repeated until the SoI is decoded while satisfying the constraint on the maximum number of cancellations. Hence, UL transmission is successful as long as one of the following events is successful
\begin{align} \nonumber \label{eq.event}
0:  & \left( \frac{Q_u h_u u^{-\alpha}}{I_{\Omega_j^0}} \geq \etat \right) \\ \nonumber
1:  & \left(\frac{Q_u h_u u^{-\alpha}}{I_{\Omega_j^0}} < \etat \right)  \bigcap  \left( \frac{X_{(1)}}{I_{\Omega_j^1}} \geq \etat \right)   \bigcap  \left( \frac{Q_u h_u u^{-\alpha}}{I_{\Omega_j^1}} \geq \etat \right) \\ \nonumber
& \vdots \\ \nonumber
N:  &  \left(\bigcap_{n = 0}^{N-1}\frac{Q_u h_u u^{-\alpha}}{I_{\Omega_j^n}}< \etat \right) \enspace\bigcap\enspace \left(\bigcap_{n=1}^{N} \frac{X_{(n)}}{I_{\Omega_j^n} } \geq \etat \right) \enspace \\ 
& \bigcap \enspace \left( \frac{Q_u h_u u^{-\alpha}}{I_{\Omega_j^N}} \geq \etat \right), 
\end{align}
where the set of interferers on subchannel $j$ after cancellation of the $n$ strongest interferers is represented by $\Omega_j^n = \Omega_j \backslash \{X_{(1)}, \hdots, X_{(n)} \}$, and $\Omega_j = \cup_{k\in \mathcal{K}}\Psi_{k,j} \backslash \lbrace u \rbrace$. The aggregate interference after cancellation is given by
\be \label{eq.aggIntSIC}
I_{\Omega_j^n} = \sum_{i = n+1}^{\infty} X_{(i)} .
\ee 
The first and third factor in the $n$-th event of \eqref{eq.event} represent outage and success for decoding the SoI when $n-1$ and $n$ interferers are canceled, respectively. The second factor in the $n$-th event of \eqref{eq.event} represents the event of successfully canceling the $n$-th interferer. 

The cancellation order is based on the received signal power and is independent of the tier to which the interferers belong. Since $X_{(i)}$ can originate from different tiers with different transmission power $Q_i$, \eqref{eq.aggIntSIC} represents the sum of order statistics, where the transmission power, fading parameter, and path loss component are random variables (r.v.'s). Recent work shows that a generic heterogeneous multi-tier network can be represented by a single-tier network where all system parameters such as the transmission power, fading parameter, and path loss exponent are set to constants, while the determinative parameter is an isotropic (possibly non-homogeneous) AP density \cite{BlaKee_PIMRC13}. For a constant path loss exponent, the isotropic density of the equivalent network reduces to a homogeneous value, as such generalizing previous results where the dispersion of the aggregate interference depends on a single moment of the transmission power and the fading distribution \cite{WinPinShe_ProcIEEE2009}. The stochastic equivalence between multi-tier and single-tier networks allows for the comparison between apparently different networks in a simplified framework. For the performance evaluation of SIC, we will relax the system model to the stochastically equivalent single-tier network with density given by $\leqi = \sum_{k \in \mathcal{K}} \lambda_k P_k^{2/\alpha}$, which follows from Campbell's theorem \cite{kingman}, and where the transmission power is equal to one.\footnote{Note that we assume in this work a constant path loss exponent $\alpha$ for all tiers.} We indicate the density of the mobile users active on channel $j$ by $\muj$. In the following sections, different lemma's are formulated, which define the probability of successfully decoding the SoI after canceling $n$ interferers and the success probability of canceling the $n$-th interferer. 

%%%%%%%%%%%%%%%%%%%%%%%%%%%%%%%%%%%%%%
\subsection{Decoding after interference cancellation}
We define the success probability of a link in an interference-limited network after successfully canceling $n$ interferers as
\begin{align} \label{eq.PsicOrdSta} \nonumber
\Psic (\etat, n) &= \Pr \left [\sum_{i=n+1}^{\infty}X_{(i)} <  h_u u^{-\alpha}/\etat \right] \\
& = \Pr \left[I_{\Omega_j^n} < h_u u^{-\alpha}/\etat\right].
\end{align}
Note that the calculation of the success probability for decoding the SoI is based on the distribution of the sum of order statistics, which requires the joint distribution of infinitely many r.v.'s. There are several possibilities to handle this problem, which have in common to limit the summation of order statistics to the $M$ strongest interferers, generally denoted as a trimmed sum. Since order statistics are mutually dependent, the cumulative distribution function (CDF) of their sum is hard to characterize. The CDF of the sum of order statistics can be found by the inverse Laplace transform of the moment generating function (MGF). It can be shown that the transformation of the order statistics of a set of exponentially distributed r.v.'s to the spacing between the order statistics results in independent r.v.'s, which alleviates the complexity to calculate the MGF \cite{DavNag_OrdSta2003}. Neglecting the topology and assuming that all interferers are at the same distance from the receiver, a closed form solution of the CDF can be found \cite{HasAloBasEbb_TWC2003}. Likewise, the difference of the square of consecutive distances follows an exponential distribution, i.e. $R_{(i)}^2 - R_{(i-1)}^2 \sim \exp(\pi \lambda)$ with $R_{(i)}^2$ the ordered Euclidean distance in $\RR^2$ of the interferers with respect to the receiver and $\lambda$ the intensity of the Poisson process \cite{Hae_TIT2005}. Hence, including the topology and neglecting fading effects, a closed form solution can be reached similarly. However, the computational complexity to calculate the CDF of a trimmed sum of order statistics including  both topology and fading is prohibitive. An alternative promising approach to characterize the sum of order statistics is to consider the asymptotic distribution of the sum $T_M(m,k) = \sum_{i = m}^{M-k}X_{(i)}$, where a fixed number $k \geq 0$ of the smallest values and a fixed number $m \geq 0$ of the largest values is trimmed, denominated as lightly trimmed sums \cite{DavNag_OrdSta2003, CsoHaeMas_AdvAppMath88, csorgo1988}. The main result of \cite{CsoHaeMas_AdvAppMath88} is based on the quantile function $Q(s)$, defined as the inverse of the CDF of the considered r.v., but no closed form solution can be found for the considered scenario. 

In our model, we include the effects of both fading and topology, yet, we assume that the order statistics are dominated by the distance. As to the received signal power in the equivalent single-tier network, the order statistics of the distance outweigh the fading effects, which have an effect on a much shorter time scale. There is a more formal motivation why we assume that the order statistics are dominated by the distance rather than the fading distribution. In the following lemma, we define the distribution of the received signal power.
\begin{lemma} \label{lem.Par}
The distribution of $Y = h X^{-\alpha}$ where $h \sim \exp(\lambda)$ follows a Pareto distribution with CDF
\be \label{eq.CDFZ}
F_Y(y) = 1 - \Gamma\left(\frac{2}{\alpha}+1\right) \frac{y^{-2/\alpha}}{R^2}
\ee
where $R$ is the maximum considered range for the position of interferers.
\end{lemma}
\begin{proof}
See Appendix \ref{app.lemPar}. 
\end{proof}
Note that $X^{-\alpha}$ follows a Pareto distribution, which belongs to the class of subexponential distributions. According to the theorem of Breiman, which states that the class of subexponential distributions is closed under the product convolution, the multiplication of $X^{-\alpha}$ with the exponential r.v. $h$ does not change the distribution \cite{CliSam_StoPro94}. Same conclusions hold for other types of fading distributions.

\begin{remark}
In the remainder, we will often make use of integrals of the form $\int 1/(1+w^{\alpha/2})\mathrm{d}w$. For the integration interval $[b, \infty)$, we define 
\begin{align} \nonumber
&C(b, \alpha) = \int_b^{\infty}  \frac{1}{1+w^{\alpha/2}}\mathrm{d}w \\
&= 2\pi/\alpha \csc(2\pi/\alpha) - b \; _2F_1(1, 2/\alpha; (2+\alpha)/\alpha; -b^{\alpha/2}),
\end{align}
where $_2F_1(.)$ is the Gaussian hypergeometric function. For special cases, we have $C(0, \alpha) = 2\pi/\alpha \csc(2\pi/\alpha)$ and $C(b, 4) = \arctan(1/b)$.
\end{remark}

\begin{lemma} \label{lemma.Psic}
A mobile user is connected to a typical AP, which has successfully canceled $n$ interferers. The success probability of UL transmission in the presence of network interference is given by \eqref{eq.Psic}, where $\Ri = \sqrt{n/\muj \pi}$ is the cancellation radius that defines the area around the victim receiver without interferers. 
\end{lemma}
\begin{proof}
Similar to the DL coverage probability derived in \cite{AndBacGan_TC2011}, we define the UL coverage probability conditioned on the distance of the intended link after successfully canceling $n$ interferers as
\addtocounter{equation}{1}
\begin{align} \label{eq.PsicConDisLap} \nonumber
\Psic(\etat, n \,|\,  u )&= \PP^{!u}\left\lbrace \frac{h_u u^{-\alpha}}{\sum_{v \in \Omega_j^n \backslash \lbrace u\rbrace } h_v v^{-\alpha}} \geq \etat \right\rbrace \\ \nonumber 
& \overset{(a)}{=}\EE_{I_{\Omega_j^n}} \left \{\PP\left[ h_u > \etat u^{\alpha} I_{\Omega_j^n}\right] \right\} \\ \nonumber
& \overset{(b)}{=}  \EE_{I_{\Omega_j^n}}\left\{ \exp (-\etat u^{\alpha} I_{\Omega_j^n}\right\} \\
& = \mathcal{L}_{I_{\Omega_j^n}}(\etat u^{\alpha}),
\end{align}
where $(a)$ holds because of Slivnyak's theorem \cite{HaeGan_NOW2009}, and where $(b)$ assumes a Rayleigh fading channel. The Laplace transform of $I_{\Omega_j^n}$ is denoted as $\mathcal{L}_{I_{\Omega_j^n}}(s)$. Similar to \cite{AndBacGan_TC2011}, the Laplace transform of the partially canceled interference is obtained by applying the probability generating functional (PGFL) and the conditional coverage probability in \eqref{eq.PsicConDisLap} can be written as
\be \label{eq.PsicConDisHelp}
\Psic(\etat, n \,|\,  u )  = \exp\left( -2 \pi \muj \int_{R_{\mathrm{I}}(n)}^{\infty} \frac{v}{1 + \frac{v^{\alpha}}{\etat u^\alpha}}\mathrm{d}v \right), 
\ee
where $R_{\mathrm{I}}(n)$ is the distance from the origin to the $n$-th interferer. By change of variable $w = v^2/(\etat^{2/\alpha}u^2)$, we can express \eqref{eq.PsicConDisHelp} as 
\begin{align} \nonumber \label{eq.PsicCond}
&\Psic(\etat, n \,|\, u ) \\ \nonumber
&= \exp\left( -\pi \muj \etat^{2/\alpha} u^2 \int_{b(u)}^{\infty} \frac{1}{1 + w^{\alpha/2}}\mathrm{d}w \right) \\
& = \exp \left(-\pi \muj \etat^{2/\alpha} C(b(u), \alpha) u^2\right),
\end{align}
where $b(u) = R_{\mathrm{I}}(n)^2/\etat^{2/\alpha}u^2$. The integration interval of the function $C(b(u), \alpha)$ depends on $R_{\mathrm{I}}(n)$, and therefore, the expectation should be taken with respect to the distance to the $n$-th interferer. From \cite{Hae_TIT2005}, the probability density function (PDF) of $R_{\mathrm{I}}(n)$ is given by
\be \label{eq.distDistrN}
f_{R_\mathrm{I}(n)}(r) = \exp(\muj \pi r^2) \frac{2(\muj \pi r^2)^n}{r\Gamma(n)} \quad .
\ee
Since the expectation can only be solved by numerical integration, we will approximate $R_{\mathrm{I}}(n)$ by the cancellation radius $\Ri = \sqrt{n/\muj\pi}$, which encloses on average $n$ mobile users such that $b(u) \approx \Ri^2/(\etat^{2/\alpha} u^2)$. As the SIC procedure cancels at each step the signal with the strongest power, we will have $u \in [\Ri, \infty)$ on average.  To guarantee the maximum average received signal power in the equivalent single-tier network, each user connects to the closest AP such that the distribution of the distance $D$ with respect to the intended base station is given by $f_D(u) = 2 \leqi \pi u \exp(-\leqi \pi u^2)$ \cite{Hae_TIT2005}.\footnote{Note that this distance distribution is only exact when the point processes of users and base stations are independent.} To find the unconditional success probability, we take the expectation over $u$ and find \eqref{eq.Psic} which concludes the proof. 
\end{proof}
%Notice that \eqref{eq.Psic} cannot be further simplified as $C(b(u), \alpha)$ is a function of $u$. 

%%%%%%%%%%%%%%%%%%%%%%%%%%%%%%%%%%%%%%
\subsection{Interference cancellation}
In the following, the success probability is derived to decode and cancel the $n$-th strongest signal, assuming that the interference from the $n-1$ strongest signals has been previously canceled. This result can be achieved by (i) building on the PGFL and distance dominated order statistics, or (ii) using the truncated stable distribution (TSD). Both approaches lead to an elegant expression of the success probability. However, the approach based on the PGFL is more general, while the approach based on the TSD gives more insight in terms of the convergence of the interference distribution. 

\subsubsection{Probability generating functional approach}
\begin{figure*}[!t]
\normalsize
\setcounter{MYtempeqncnt}{\value{equation}}
\setcounter{equation}{21}
\be \label{eq.Pssic}
\PP_{\mathrm{s, SIC}}(\etat, n) \approx \PP_\s(\etat) + \sum_{i=1}^{N} \left(\prod_{n = 0}^{i-1}(1-\PP_{\mathrm{s, IC}}(\etat, R_{\mathrm{I},n}))\right) \left( \prod_{n=1}^{i} \PP_{\mathrm{s, can}}(\etat, n) \right) \PP_{\mathrm{s, IC}}(\etat, R_{\mathrm{I}, i}) 
\ee
\setcounter{equation}{\value{MYtempeqncnt}}
\hrulefill
\end{figure*}
\begin{lemma} \label{lemma.Pscan}
Considering a typical AP that successfully canceled the $n-1$ strongest signals, the success probability to decode the $n$-th strongest signal is given by
\be \label{eq.Pscan}
\Pscan(\etat, n) = \frac{1}{(1 + \etat^{2/\alpha} C(1/\etat^{2/\alpha}, \alpha))^{n}} \quad .
\ee
\end{lemma}
\begin{proof}
In the SIC scheme, interferers are canceled according to descending received signal power. We consider the order statistics $X_{(i)}$ of the  signal power to define the probability of successfully decoding the $n$-th strongest signal. After successfully decoding and subtracting $n-1$ signals from the received signal, the success probability can be written as
\be
\Pr\left[ \frac{X_{(n)}}{\sum_{i \in \Omega_j^{n}}X_{(i)}} \geq \etat \right] = \Pr \left[\sum_{i = n+1}^{\infty} X_{(i)} \leq X_{(n)}/\etat \right] .
\ee  
Since it is unwieldy to characterize the distribution of the sum of order statistics including both geometry and fading, we assume that the order statistics are dominated by the distance, such that  $X_{(j)} \geq X_{(i)}$ with $i < j$ is equivalent to $v_j \leq v_i$. When the SoI is the strongest signal from the received multi-user signal corresponding to distance $v_j$, then the distance of the remaining interferers will be in the interval $[v_j, \infty)$. This approximation leads to a remarkable simplification for the calculation of the probability of successfully decoding the $n$-th interferer conditioned on $v_n$, which can then be expressed as
\begin{align} \nonumber \label{eq.PscanCond}
&\Pscan(\etat, n \,|\, v_n) = \Pr\left( \frac{X_{(n)}}{I_{\Omega_j^n}} \geq \etat \right) \\ \nonumber
& = \exp\left( -\pi\muj \etat^{2/\alpha} v_n^2 \int_{R_{\mathrm{I}}(n)^2/(\eta^{2/\alpha}v_n^2)}^{\infty} \frac{1}{1+w^{\alpha/2}} \mathrm{d}w \right) \\
& = \exp\left( -\pi\muj \etat^{2/\alpha} v_n^2 C(1/\etat^{2/\alpha}, \alpha) \right)  \quad .
\end{align}
Note that the residual interferers are located outside the circular area with radius $v_n$, where $R_{\mathrm{I}}(n) = v_n$ and the function $C(1/\etat^{2/\alpha}, \alpha)$ is independent of $v_n$. Using \eqref{eq.distDistrN} and \eqref{eq.PscanCond}, we get
\begin{align} \nonumber \label{eq.Pscan_ave}
&\Pscan(\etat, n) = \int_0^\infty \exp\left(-\pi \muj \etat^{2/\alpha} C(1/\etat^{2/\alpha}, \alpha) v_n^2 \right) \\ \nonumber
& \hspace{7em} \exp(-\pi \muj v_n^2) \frac{2(\pi \muj v_n^2)^n}{v_n \Gamma(n)}\mathrm{d}v_n \\ \nonumber
 &= \frac{(\pi \muj)^n}{\Gamma(n)} \int_0^\infty \exp\left( -\pi \muj(1+ \etat^{2/\alpha} C(1/\etat^{2/\alpha}, \alpha))v_n^2 \right) \\ \nonumber 
& \hspace{18em} v_n^{2n-2} \mathrm{d}v_n^2  \\ \nonumber
&= \frac{(\pi \muj)^n}{\Gamma(n)} \int_0^\infty \exp\left( -\pi \muj(1+ \etat^{2/\alpha} C(1/\etat^{2/\alpha}, \alpha))w \right) \\ \nonumber
& \hspace{18em} w^{n-1} \mathrm{d}w \\ \nonumber 
&= \frac{(\pi \muj)^n}{\Gamma(n)}  (\pi \muj (1 + \etat^{2/\alpha} C(1/\etat^{2/\alpha}, \alpha))^{-n} \Gamma(n) \\
&= \frac{1}{(1 + \etat^{2/\alpha} C(1/\etat^{2/\alpha}, \alpha))^{n}} .
\end{align}
\end{proof}
It is worthwhile to note that \eqref{eq.Pscan_ave} is independent of $\muj$. This is in line with \cite{ZhaHae_Globecom2012, ZhaHae_Asilomar13} where the authors prove that the probability to successfully cancel at least  $n$ interferers is scale invariant with respect to the density as long as the analysis is restricted to the power-law density case.  

\subsubsection{Truncated stable distribution approach}
In the unbounded path loss model $g_{\alpha}(x) = \|x\|^{-\alpha}$, the aggregate interference power generated by the interfering nodes scattered over $\RR^2$ can be modeled by a skewed stable distribution $I_{\Omega_j}  \sim \mathcal{S}(\alpha_I = 2/\alpha, \beta = 1, \gamma = \pi \lambda C^{-1}_{2/\alpha}\EE\lbrace P_i^{2/\alpha} \rbrace)$ \cite{WinPinShe_ProcIEEE2009}. Since the singularity at $0$ in the unbounded path loss model can have significant effects on the interference distribution \cite{InaChiaPooWic_JSAC09}, a bounded path loss model based on the truncated stable distribution (TSD) has been proposed to avoid the singularity for zero distance by restricting the interferers to a ring structure with a minimum range $d_{\text{min}}$ and maximum range $d_{\text{max}}$ \cite{RabQueShiWin_JSAC2011}. The ring structure can be applied similarly to represent the guard zone around the victim receiver in the SIC scenario, where the inner radius corresponds to the interference cancellation radius and $d_{\text{max}} = \infty$.  In the bounded path loss model, the aggregate interference power is distributed according to a skewed TSD with CF given by \cite{RabConWin_Globecom2011}
\be \label{eq.TSD}
\psi_{I_{\Omega_j^n}}(j\omega) = \exp\left( \gamma' \Gamma(-\alpha_I)\left[ (g - j\omega)^{\alpha_I} -g^{\alpha_I} \right] \right), 
\ee
where the parameters $\alpha_I$, $g$, and $\gamma'$ determine the shape of the TSD. 
\begin{lemma} \label{lem.PscanTSD}
For $\alpha = 4$, the probability to decode and cancel the $n$-th strongest signal after $n-1$ successful cancellations is given by
\be \label{eq.PscanTSD}
\Pscan(\etat, n|\alpha =4) = \frac{1}{(\sqrt{9/4+3\etat}-1/2)^n} .
\ee
\end{lemma}
\begin{proof}
See Appendix \ref{app.lemPscanTSD}.
\end{proof}

Note that \eqref{eq.PscanTSD} has the same structure as \eqref{eq.Pscan_ave} and similarly shows that the probability of canceling the $n$-th interferer is independent of the interferer density. The use of the TSD approach is beneficial since it can provide insight how fast the aggregate network interference converges to a Gaussian distribution by calculating the kurtosis after canceling $n$ interferers. Conditioned on the exclusion radius $r$, we have 
\be \label{eq.kurDis}
\gamma_2 = \frac{\kappa(4)}{\kappa(2)^2} = \frac{6(\alpha - 1)^2}{(2\alpha - 1)}\frac{1}{\muj \pi r^2},
\ee
where $\gamma_2 = 0$ represents the case of a normal distribution. By averaging over $r$, the kurtosis after canceling $n$ interferers is given by
\begin{align} \nonumber
\gamma_2(\alpha, n) &= \frac{6(\alpha - 1)^2}{(2\alpha - 1)} \int_0^\infty \frac{1}{\muj \pi r^2} \exp(-\muj \pi r^2) \frac{2(\muj \pi r^2)^n}{r\Gamma(n)}\mathrm{d}r \\ \nonumber
%& = \frac{6(\alpha - 1)^2}{(2\alpha - 1)} \frac{1}{\Gamma(k)} \int_0^\infty \exp(-\muj \pi r^2) 2 (\muj \pi)^{k-1} r^{2k-3} \mathrm{d}r \\ \nonumber
&= \frac{6(\alpha - 1)^2}{(2\alpha - 1)} \frac{1}{\Gamma(n)} \int_0^\infty \exp(-\muj \pi r^2) (\muj \pi r^2)^{n-2} \mathrm{d} \muj \pi r^2 \\ \nonumber
&= \frac{6(\alpha - 1)^2}{(2\alpha - 1)} \frac{\Gamma(n-1)}{\Gamma(n)} \\
& =  \frac{6(\alpha - 1)^2}{(2\alpha - 1)} \frac{1}{n-1} .
\end{align} 
This expression yields useful insight in the convergence rate of the aggregate network interference to a Gaussian distribution as a function of the number of canceled interferers.

\subsection{Success probability with SIC}
\begin{theorem} \label{theo.Pssic}
The coverage probability $\PP_{\mathrm{s, SIC}}$ for a receiver that applies SIC with a maximum of $n$ interferers being canceled is given by \eqref{eq.Pssic}, where $\Psic(\etat, \Ri)$ and $\Pscan(\etat, n)$ are given in \eqref{eq.Psic} and \eqref{eq.Pscan}, respectively.  
\end{theorem}
\begin{proof}
The proof follows directly from the definition of the sequence of events in \eqref{eq.event} and the results of $\Psic$ and $\Pscan$ in Lemma \ref{lemma.Psic} and \ref{lemma.Pscan}. 
\end{proof}

%\begin{remark} So far, we considered perfect interference cancellation in our framework, which has also been adopted in \cite{HuaLauChe_JSAC2009, HasAloBasEbb_TWC2003}. Interference cancellation is however inevitably imperfect. In \cite{BenCarGag_TC2010}, the strongest interferer within the group of partially canceled and uncanceled interferers are considered to bound the success probability. Also \cite{WebAnd_NOW2012} consider partial interference cancellation based on the void probability of the dominant and non-dominant interferers. We verify if we our framework can accommodate for residual interference, by introducing the factor $\kappa$ which determines the residual interference of cancelable nodes
%\begin{align}\nonumber
%\Psmo(\eta_\mm|r) &= \EE \left[ \exp(-\eta_\mm r^{\alpha} (\kappa I_{\Omega\backslash \Omega_K} + I_{\Omega_K}))  \right] \\ 
%&\overset{(a)}{\geq} \mathcal{L}_{I_{\Omega\backslash \Omega_K}}(\eta_\mm r^\alpha \kappa) \mathcal{L}_{I_{\Omega_K}}(\eta_\mm r^{\alpha})
%\end{align}
%where the inequality $(a)$ holds due to the positive correlation between $I_{\Omega\backslash \Omega_K}$ and $I_{\Omega_K}$ and Harris' inequality.
%\end{remark}

We provide now some numerical results that illustrate the effectiveness of SIC for the association policy based on the maximum average received signal power. We consider a network of APs arranged over a two-dimensional plane with density $\lm = 10^{-4} \enspace /\text{m}^2$. We assume a fully loaded network, where each cell allocates at a given time every subchannel to an active user. Hence, the density of mobile users on subchannel $j$ is given by $\muj = 10^{-4}/\text{m}^2$. 
In Fig. \ref{fig.Ps_can_compare}, $\Pscan$ is depicted for different values of the threshold as a function of the order of the canceled interferer. The figure illustrates that the probability to cancel the $n$-th interferer decreases quickly with $n$ and with increasing target SIR. Simulation results are added to validate the model. When the received signal power is ordered only with respect to the distance, the simulations coincide with the analytical results based on the PGFL approach and the TSD approach. Moreover, a good agreement between analysis and simulations is achieved even when the ordering is performed based on the joint effect of distance and fading. From the numerical results depicted in Fig. \ref{fig.Ps_can_compare}, we observe that the approximation deteriorates for lower values of the SIR threshold. As the SIR threshold decreases, we resort to values drawn from the central part of the PDF of the SINR, whereas the theorem of Breiman, necessary to assume that the order statistics are dominated by the distance, holds for the tail of the distribution.  
\begin{figure}[t!]
\centering
\psfrag{number of interferers cancelled}[][][0.75]{$n$}
\psfrag{Pscan}[][][0.75]{$\Pscan$}
\includegraphics[width = 1 \linewidth] {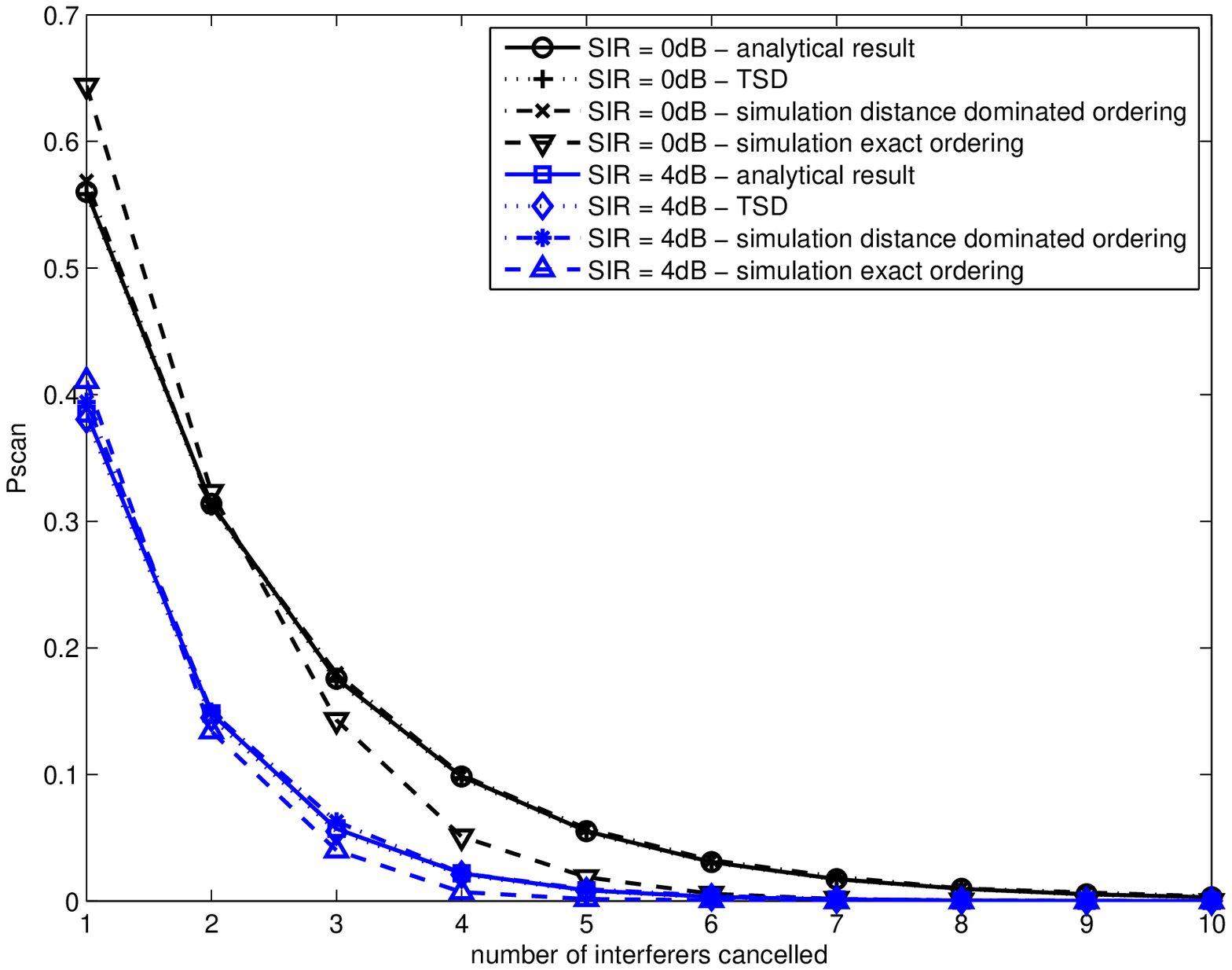}
\caption{Comparison of $\Pscan$ using analytical and simulation results.}
\label{fig.Ps_can_compare}
\end{figure}

Figure \ref{fig.PssicVsHuang} illustrates the success probability including SIC as a function of the SINR target for different values of the maximum number of cancellations. To validate our analysis, we compare the results with bounds that have been proposed in \cite{HuaLauChe_JSAC2009} for the scenario of spectrum sharing between cellular and mobile ad hoc networks (MANET), which can be shaped to represent a single-tier cellular network. In this work, the authors present bounds based on the separation of interferers into groups of strong and weak interferers where each strong interferer alone can cause outage, accounting both for the effects of the spatial distribution of the nodes and the fading affecting each link. Interference cancellation is performed in descending order or received signal power, and the received power of each interferer intended for cancellation must exceed the SoI signal power multiplied with a factor $\kappa > 1$. From Fig. \ref{fig.PssicVsHuang}, we observe that the curves derived by our analysis strictly fall within the bounds proposed in \cite{HuaLauChe_JSAC2009}. Furthermore, we observe a modest improvement in the success probability when SIC is applied for threshold values lower than $2$ dB, whilst for higher threshold values this improvement is negligible. The numerical results illustrate that the cancellation of the first order interferer has a sensible effect on the receiver performance, while the cancellation of higher order interferers yields a marginal improvement of the success probability.
\begin{figure}[t!]
\centering
\psfrag{Ps with SIC}[][][0.75]{$\PP_{\mathrm{s, SIC}}$}
\psfrag{Threshold (dB)}[][][0.75]{$\etat$  [dB]}
\psfrag{SIC - k = 0}[][][0.5]{SIC - $ k = 0$}
\psfrag{SIC - k = 1}[][][0.5]{SIC - $ k = 1$}
\psfrag{SIC - k = 2}[][][0.5]{SIC - $ k = 2$}
\psfrag{SIC - k = 3}[][][0.5]{SIC - $ k = 3$}
\psfrag{SIC - k = 4}[][][0.5]{SIC - $ k = 4$}
\psfrag{SIC - k = 5}[][][0.5]{SIC - $ k = 5$}
\psfrag{UB}[][][0.5]{UB}
\psfrag{LB}[][][0.5]{\,LB}
\includegraphics[width = 1 \linewidth]{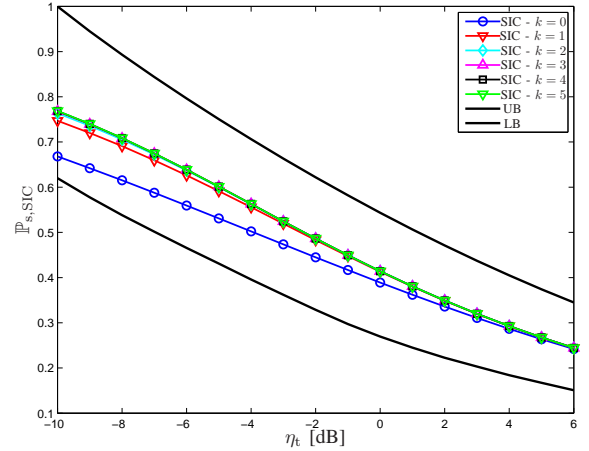}
\caption{Coverage probability in the presence of SIC for different values of the maximum number of cancellations. The blue curve represents the success probability when no IC technique is used.}
\label{fig.PssicVsHuang}
\end{figure}

The results presented in this section provide a guideline for the SIC computational requirements by investigating the performance gain from the  cancellation of $n$ interferers. Although applicable for UL and DL transmissions, SIC is particularly attractive in UL since it harnesses the processing power of access points to cancel strong interfering signals from nearby transmitters. However, as only the first cancellation has a significant effect on the performance, the computational requirements related to SIC are limited and hence, SIC qualifies also for DL transmissions. Following the association policy where the nodes connect to the AP offering the highest average received signal power, the overall performance gain of SIC is however modest for realistic values of the target SINR. In the following section, we present several association policies where the performance gain of SIC is far more appreciable. 

%%%%%%%%%%%%%%%%%%%%%%%%%%%%%%%%%%%%%%
\section{Association policies and SIC gains} \label{sec.SICBen}
In this section, we analyze multi-tier networks with different association policies. Since we aim to deepen the understanding of heterogeneous networks using concepts such as coverage area and load, we will not resort to the single-tier stochastic equivalent of the network if this results in a loss of physical insight related to the differences between the tiers. %By means of several case studies, we show that SIC can be relevant in scenarios where a link is established that does not correspond with the maximum average received signal power.

%%%%%%%%%%%%%%%%%%%%%%%%%%%%%%%%%%%%%%
\subsection{Minimum load association policy}
Considering fairness between users, in case of data sensitive applications it can be preferential to connect to the AP with the lowest load, rather than to the AP that offers the highest SIR. The same observation holds for networks that apply a load balancing policy and where users are actively transferred to lightly loaded APs different from the AP of their own Voronoi cell \cite{YeRonCheAlsCarAnd_TWC13}. In the following, we consider the association policy where a user connects to the AP with the lowest load for a given connectivity range $R_{\text{con}}$ with respect to the user. In this scenario, the performance metric of interest is the rate per user, which reflects the quality of service (QoS) and depends on the AP load, defined as the number of users $M$ connected to the AP. This scenario leads to interesting trade-offs between APs where the loss of SIR can be compensated by the gain of available resource blocks per user. We consider a single tier network and we model explicitly the load of the APs by considering the marked PPP $\tilde{\Phi} = \{ (X_i, L_i)\, | \, X_i \in \Phi(\lambda), \, L_i \sim F_L(l)\}$, with $L_i$ the load of $X_i$ and $F_L(l)$ the load distribution. We consider a typical user at the origin of the Euclidean plane and we compare the performance of the max-SIR association policy with the minimum load association policy for DL transmissions in terms of rate per user. Let $\mathcal{R}$ denote the rate per user and we define the coverage probability as the  CCDF of the rate $\Prc(\rho) = \PP[\mathcal{R} > \rho]$ \cite{SinDhiAnd_TWC13}, which is given by 
\addtocounter{equation}{1}
\begin{align} \nonumber
\Prc(\rho) &= \Pr\left[\frac{1}{M}\log(1 + \text{SIR}) > \rho \right] \\
&= \Pr[\text{SIR} > 2^{M\rho}-1] = \EE_M[\PP_\s(2^{M\rho}-1)] \quad.
\end{align}
To calculate the coverage probability, we need to characterize the distribution of $M$. The load of a cell depends on the area distribution of the Voronoi cell, for which an approximation has been proposed in \cite{FerNed_PhyA}. Using this approximation, the probability mass function of $M$ is given by \cite{SinDhiAnd_TWC13}
\begin{align} \nonumber
f_M(m) &= \Pr[M = m] \\
& = \frac{3.5^{3.5}}{m!} \frac{\Gamma(m+4.5)}{\Gamma(3.5)}\left(\frac{\muj}{\lambda} \right)^m \left( 3.5 + \frac{\muj}{\lambda}\right)^{-(m+4.5)} .
\end{align}
In case of the max-SIR policy, the coverage probability conditioned on the number of associated users is given by
\begin{align} \nonumber
&\PrcMax(\rho | M) \\ \nonumber
&= \int_0^\infty 2 \lambda \pi  r \exp(-\pi \lambda r^2 (1 + \varsigma^{2/\alpha}C(1/\varsigma^{2/\alpha},\alpha)))\mathrm{d}r \\
&= \frac{1}{1 + \varsigma^{2/\alpha}C(1/\varsigma^{2/\alpha},\alpha)},
\end{align}
where $\varsigma = 2^{M\rho}-1$. 
Deconditioning over $M$, the rate coverage can be written as
\be
\PrcMax(\rho) = \sum_{m \geq 0}f_M(m)\PP_\s(2^{\rho(m+1)}-1),
\ee
where the load of the cell under consideration includes the admitted user.
\begin{lemma}
For a typical user that connects to the AP with the lowest load within the range $R_{\mathrm{con}}$, the coverage probability is given by
\be
\PrcMin(\rho) =  \sum_{m \geq 0 } f_{M_{(1)}}(m) \PP_{\s}(2^{\rho(m+1)}-1),
\ee
where $f_{M_{(i)}}(m)$ represents the \ac{PMF} of the $i$-th order statistic of the load.
\end{lemma}
\begin{proof}
For the minimum load association policy, the typical user is appointed to the AP with the lowest load which is uniformly distributed over $b(0, R_{\text{con}})$ with distance distribution $f_R(r) = 2r/R_{\text{con}}^2$. We assume that there are $N = \lfloor \lambda \pi R_{\text{con}}^2 \rfloor$ APs within the connectivity range. The coverage probability for the minimum load scheme conditioned on the load is given by
\begin{align} \nonumber \label{eq.PrcMinM}
\PrcMin(\rho | M) &= 1/R_{\text{con}}^2 \int_0^{R_{\mathrm{con}}^2}\exp(\pi\lambda \varsigma^{2\alpha}C(0,\alpha)r^2)2r\mathrm{d}r \\
&= \frac{1- \exp(-\pi \lambda \varsigma^{2/\alpha}C(0,\alpha)R_{\text{con}}^2)}{\pi \lambda \varsigma^{2/\alpha}C(0,\alpha)R_{\text{con}}^2} .
\end{align}
The $i$-th order statistic of the load is given by \cite{DavNag_OrdSta2003}
\be
f_{M_{(i)}}(m)  = \frac{1}{\mathcal{B}(i,M-i+1)}\int_{F_M(m-1)}^{F_M(m)}w^{i-1}(1-w)^{m-i}\mathrm{d}w ,
\ee
where $\mathcal{B}(a,b)$ represents the beta function. Deconditioning \eqref{eq.PrcMinM} with respect to the first order statistic of the load, the proof is concluded.
\end{proof}

%%%%%%%%%%%%%%%%%%%%%%%%%%%%%%%%%%%%%%
\subsection{Maximum instantaneous received signal power}
Connecting to the AP which yields instantaneously the highest SIR can be of interest for a mobile node to obtain the maximum data rate \cite{DhiGanBacAnd_JSAC12}, or to reduce the local delay $\tau = \EE\left\lbrace 1/{\PP_\s(\etat)} \right \rbrace$, which is defined as the mean time until a packet is successfully received \cite{Hae_TIT13}. In the following, we will evaluate how SIC can affect the UL success probability when the maximum instantaneous SIR policy is applied. 
%Without SIC and assuming that all thresholds $\gamma_i > 1, \forall i$, it can be demonstrated that the success probability is given by 
%\be
%\frac{1}{C(0, \alpha)} \frac{\sum_{k\in \mathcal{M}}\lambda_i \gamma_i^{-2/\alpha}Q_i^{2/\alpha}}{\sum_{k\in \mathcal{M}}\mu_k Q_k^{2/\alpha}}
%\ee
%where $\mu_k = p_k \mu $ represents the portion of the users which is associated with the $k$-th tier. For the max-SINR policy, the association probability is given by 
%\be
%p_k = \frac{\lambda_k P_k^{2/\alpha} \gamma_k^{-2/\alpha}}{\sum_{i \in \mathcal{M}} \lambda_i P_i^{2/\alpha} \gamma_i^{-2/\alpha}}
%\ee

In this scenario, we assume that (i) the current network state and load per AP is based on mobile node connecting to the AP which provides the maximum DL average received signal power, (ii) all tiers have equal SIR target values, and (iii) the network operates in time division duplex (TDD) mode, such that we can assume channel reciprocity, i.e. the maximum SIR in DL is also the max SIR in UL. We first determine the node densities $\mu_k$ connected to tier $k$, based on the maximum DL average received signal power. This approach is realistic since it does not require extra signaling from the APs. The user density is given by $\mu_k = \pak \mu$, where $\pak$ is the association probability to tier $k$ given by
\begin{align} \nonumber \label{eq.pa_DL}
&\pak = \Pr \left \lbrace \bigcap_{i \neq k} \left( \max_{\Phi_k} \text{SIR}_j > \max_{\Phi_i} \text{SIR}_j \right) \right\rbrace \\ \nonumber
& = \prod_{i \neq k} \Pr\left \lbrace P_k x_k^{-\alpha} > P_i x_i^{-\alpha} \right\rbrace \\ \nonumber
& = \prod_{i \neq k} \Pr\left \lbrace x_i > x_k (P_i/P_k)^{1/\alpha} \right\rbrace \\ \nonumber
& = \int_0^{\infty} \exp(- \pi x_k^2 \sum_{i \neq k} \lambda_i(P_i/P_k)^{2/\alpha}) 2 \pi \lambda_k x_k \exp(-\lambda_k \pi x_k^2) \mathrm{d}x_k \\
& = \frac{\lambda_k}{\sum_{i \in \mathcal{K}}\lambda_i \left( \frac{P_i}{P_k}\right)^{2/\alpha}}, 
\end{align}
where $x_i$ represents the distance to the closest AP of tier $i$ and we consider that the maximum SIR corresponds to the maximum received signal power. In the following lemma, we define the outage probability for UL transmissions in a multi-tier heterogeneous network without IC capabilities.
%Not ignoring the effect of fading and for the UL channel, we can define the association probability $\pa^{(i)} = \PP_{j\neq i} (\max \quad \text{SIR}_j < \max \quad \text{SIR}_i)$ with 
%\be
%F_i(\gamma) = \exp \left( \frac{-\lambda_i Q_i}{\gamma^{2/\alpha} C(0,\alpha) \sum \mu_k Q_k^{2/\alpha}} \right)
%\ee
%and
%\be
%f_i(\gamma) = \exp \left( \frac{-\lambda_i Q_i}{\gamma^{2/\alpha} C(0,\alpha) \sum \mu_k Q_k^{2/\alpha}} \right) \frac{2}{\alpha}  \frac{-\lambda_i Q_i}{ C(0,\alpha) \sum \mu_k Q_k^{2/\alpha}}\gamma^{-2/\alpha-1}
%\ee
%such that
%\begin{align} \nonumber
%\pa^{(i)} & = \int_0^\infty  \exp \left( \frac{-\lambda_j Q_j}{\gamma^{2/\alpha} C(0,\alpha) \sum \mu_k Q_k^{2/\alpha}} \right) \exp \left( \frac{-\lambda_i Q_i}{\gamma^{2/\alpha} C(0,\alpha) \sum \mu_k Q_k^{2/\alpha}} \right) \\ \nonumber
%& \quad \times \frac{2}{\alpha}  \frac{-\lambda_i Q_i}{ C(0,\alpha) \sum \mu_k Q_k^{2/\alpha}}\gamma^{-2/\alpha-1}  d\gamma \\ \nonumber
%& = \int_0^\infty \exp\left(-u (1 + \lambda_j Q_j^{2/\alpha}/ \lambda_i Q_i^{2/\alpha}) \right) \mathrm{d}u \\
%& = \prod_{j \neq i} \frac{1}{1 + \frac{\lambda_j Q_j^{2/\alpha}}{\lambda_i Q_i^{2/\alpha}}}
%\end{align}
%such that the structure is equal to \eqref{eq.pa_DL} where the association probability depends on the ratio of the densities and the transmission powers, respectively. \mynote{We should consider also DL SIR distribution including fading}
\begin{lemma} \label{lem.PoutMultiMax}
If a typical user connects to the AP which yields the maximum instantaneous SIR, then the outage probability is given by 
\be \label{eq.PoutMaxSIR}
\Pout(\etat, \boldsymbol{\lambda}, \boldsymbol{\mu}, \mathbf{Q}) = \exp\left( \frac{- \sum_{j=1}^{K}\lambda_j Q_j^{2/\alpha}}{\etat^{2/\alpha} C(0,\alpha) \sum_{i=1}^K \mu_i Q_i^{2/\alpha}}\right) .
\ee
\end{lemma}
\begin{proof}
For  mobile nodes applying the maximum instantaneous SIR association policy, the outage probability is given by
\begin{align} \nonumber \label{eq.PoutmaxSIR}
&\Pout(\etat, \boldsymbol{\lambda}, \boldsymbol{\mu}, \mathbf{Q}) = \Pr \lbrace \underset{k \in \mathcal{K},x_{k,i} \in \Phi_k }{\max} \text{SIR}_j(u \to x_{k,i}) < \etat \rbrace \\ \nonumber
& =  \underset{k \in \mathcal{K},x_{k,i} \in \Phi_k }{\bigcap} \Pr\lbrace\text{SIR}_j(u \to x_{k,i}) < \etat \rbrace \\ \nonumber
 &\overset{(a)}{=}  \bigcap_{k \in \mathcal{K}} \EE_{\Phi_k} \left \lbrace \prod_{x_{k,i} \in \Phi_k} \PP^{!u} \left( \frac{Q_k h_u u^{-\alpha}}{I_{\Omega_j}} < \etat \right) \right \rbrace \\
 & = \prod_{k \in \mathcal{K}} \Pout^{(k)}(\etat) \, ,
\end{align}
where $\Pout^{(k)}$ is the outage probability of a typical user connected to the AP of tier $k$ that yields the maximum instantaneous SIR. Note that in $(a)$ we assume that all $\text{SIR}_j(u\to x_{k,i})$ are independent.\footnote{In \cite{DhiGanBacAnd_JSAC12}, the authors limit the range of the SIR threshold to $\etat > 1$ to avoid the dependence, while an exact approach is proposed in \cite{Muk_JSAC12} where the joint SINR distribution is presented.} We can further develop $\Pout^{(k)}$ as
\begin{align} \nonumber \label{eq.PoutTierk}
&\Pout^{(k)}(\etat, \lambda_k, \boldsymbol{\mu}, \mathbf{Q}) \\ \nonumber
&= \EE_{\Phi_k} \left[ \prod_{x_{k,i} \in \Phi_k} 1 - \EE_{I_\Omega}\left[ \exp \left(-\etat \frac{u^{\alpha}}{Q_k} I_{\Omega_j} \right)\right] \right] \\ \nonumber
& = \exp\left(- 2 \pi \lambda_k \int_0^\infty \EE_{I_{\Omega_j}}\left[ \exp\left(-\etat \frac{u^\alpha}{Q_k} I_{\Omega_j} \right)\right] u \mathrm{d}u \right)\\ \nonumber
& = \exp \Big(-2 \pi \lambda_k \\ \nonumber
&\hspace{1 em} \times \int_0^\infty \exp\left(-\pi \etat^{2/\alpha} C(0, \alpha) u^2 \sum_{i = 1}^K \mu_i \left( \frac{Q_i}{Q_k}\right)^{2/\alpha} \right) u \mathrm{d}u \Big) \\
& =  \exp\left(\frac{-\lambda_k }{\etat^{2/\alpha}C(0, \alpha) \sum_{i=1}^{K}\mu_i (Q_i/Q_k)^{2/\alpha}}\right) .
\end{align}
Combining \eqref{eq.PoutmaxSIR}  and \eqref{eq.PoutTierk}, the proof is concluded.
\end{proof}
From Lemma \ref{lem.PoutMultiMax}, we can conclude that the density of the superposition of network nodes with different transmit powers is equal to a weighted sum of the densities $\tilde{\mu}_k = \sum_{i=1}^{K}\mu_i (Q_i/Q_k)^{2/\alpha}$. 
\begin{lemma}
The success probability for UL transmissions with SIC is given by
\begin{align} \nonumber
&\Ps^{\mathrm{SIC}}(\etat, \boldsymbol{\lambda}, \boldsymbol{\mu}, \mathbf{Q}) = 1 - \Pout(\etat, \boldsymbol{\lambda}, \boldsymbol{\mu}, \mathbf{Q}) \\
& \times \prod_{k \in \mathcal{K}} \left \lbrace \exp\left( - 2 \pi \lambda_k \int_0^\infty \PP_{\mathrm{s,SIC}}^{\mathrm{gain}}(\gamma, N|u)u \mathrm{d}u \right) \right \rbrace
\end{align}
where
\begin{align} \nonumber
&\PP_{\mathrm{s,SIC}}^{\mathrm{gain}}(\gamma, N|u) = \\ \nonumber
&\sum_{i=1}^{N} \left(\prod_{n=0}^{i-1} \left[1 - \exp(-\pi \etat^{2/\alpha} C(R_{\mathrm{I},n}^2/\etat^{2/\alpha}u^2, \alpha)\tilde{\mu}_j u^2)\right]\right) \\ \nonumber 
& \times \left( \prod_{n=1}^i \PP_{\mathrm{s, can}}(\etat, n) \right)\\
& \times \exp(-\pi \etat^{2/\alpha}C(R_{\text{I},i}^2/\etat^{2/\alpha}u^2, \alpha)\tilde{\mu}_j u^2) \,.
\end{align}
\end{lemma}
\begin{proof}
When a typical user is associated to the AP corresponding to the highest instantaneous SIR, we can write 
\be \label{eq.PoutHelp1}
\Pout^{\text{SIC}}(\etat, \boldsymbol{\lambda}, \boldsymbol{\mu}, \mathbf{Q}) = \bigcap_{k \in \mathcal{K}} \EE_{\Phi_k} \left \lbrace  \prod_{x_{k,i} \in \Phi_k} \PP\left(\frac{Q_k h_u u^{-\alpha}}{I_{\Omega_j^{n}}}  < \etat \right) \right \rbrace
\ee 
where $\PP({Q_k h_u u^{-\alpha}}/I_{\Omega_j^{n}}<\etat)$ is the outage probability conditioned on the distance $u$ and $n$ interferers being canceled. From \eqref{eq.PoutHelp1}, we can rewrite as
\begin{align} \label{eq.PoutSIC} \nonumber
&\Pout^{\text{SIC}}(\etat, \boldsymbol{\lambda}, \boldsymbol{\mu}, \mathbf{Q}) = \\
& \prod_{k \in \mathcal{K}} \left \lbrace \exp\left( - 2 \pi \lambda_k \int_0^\infty \underbrace{\EE_{I_{\Omega_j^{n}}}[\exp(-\etat u^{-\alpha} I_{\Omega_j^{n}}/Q_k)]}_{\PP_{\text{s,SIC}}^{(k)}(\etat, n|u)}u \mathrm{d}u \right) \right \rbrace
\end{align}
where the typical user is assumed to be at the origin and the APs at a distance $u$ from the origin. Still, the interference cancellation radius for each AP is considered with respect to the AP under consideration, which does not induce any problems since we consider motion-invariant point processes. Therefore, the success probability $\PP_{\text{s,SIC}}^{(j)}(\etat, n|u)$ referenced to tier $j$ is invariant under translation. The success probability including SIC conditioned on the distance can further be written as 
\begin{align} \nonumber \label{eq.PsSICInsConU}
&\PP_{\text{s,SIC}}^{(j)}(\etat, n|u) \\ \nonumber
&=\PP_\s(\etat|u) + \sum_{i=1}^{N} \left(\prod_{n = 0}^{i-1}(1-\PP_{\text{s, IC}}(\etat, R_{\mathrm{I},n}|u))\right) \\ \nonumber
& \hspace{1em} \times \left( \prod_{n=1}^{i} \PP_{\mathrm{s, can}}(\etat, n) \right) \PP_{\mathrm{s, IC}}(\etat, R_{\mathrm{I}, i}|u) \\ \nonumber
&= \exp(-\pi \etat^{2/\alpha}C(0,\alpha) \tilde{\mu}_j u^2)  \\ \nonumber
&+ \sum_{i=1}^{N} \left(\prod_{n=0}^{i-1} \left[1 - \exp(-\pi \etat^{2/\alpha} C(R_{\mathrm{I},n}^2/\etat^{2/\alpha}u^2, \alpha)\tilde{\mu}_j u^2)\right]\right)\\ 
&\times  \left( \prod_{n=1}^i \PP_{\mathrm{s, can}}(\etat, n) \right)\exp(-\pi \etat^{2/\alpha}C(R_{\text{I},i}^2/\etat^{2/\alpha}u^2, \alpha)\tilde{\mu}_j u^2) .
\end{align}
Note that in \eqref{eq.PsSICInsConU} the single-tier equivalent network referenced to the transmission power of tier $j$ is used. Substituting \eqref{eq.PsSICInsConU} in \eqref{eq.PoutSIC}, the proof is concluded. 
\end{proof}
Considering network infrastructure with SIC capabilities, the typical user will connect to the AP that offers the highest instantaneous SIR after canceling $n$ interferers. The maximum instantaneous SIR association policy that accommodates for SIC, does not necessarily connect the mobile node to the AP which yields the maximum average SIR nor to the closest one.  

\subsection{Range expansion}
While the higher tiers in a multi-tier network are intended to offload data traffic from the macrocell network, this target is impeded considerably due to the relatively small coverage area of the higher tiers, which are usually denoted as small cells. To encourage users to connect to the small cells, range expansion has been proposed which applies an association policy based on a biased received signal power \cite{LopGuvRocKouQueZha_WComMag11}. Although range expansion mitigates the UL cross-tier interference, users in DL experience bad signal conditions in the range expanded areas since they are not connected to the base station that provides the highest average SIR. It is therefore meaningful to study the benefit of SIC in DL for those users located in the range expanded areas (REAs). To calculate the success probability of the users belonging to the REA, we need to define the distance distribution of these users with respect to the serving AP. The following lemma is an extension of \cite{TanPenHonXue_WCL13} for a $K$-tier network. \begin{lemma}
Let $B = \{b_k\}$ be the set of biases corresponding to each tier. The distance distribution of users located in the REA to the serving AP is given by
\begin{multline}
f_{X_k^{\mathrm{(RE)}}}(x) =  \frac{2 \pi \lambda_k}{\pare} x \Bigg[ \exp \left( -\pi \sum_{i \in \mathcal{K}} \lambda_i \left( \frac{P_i}{P_k} \frac{b_i}{b_k} \right)^{2/\alpha} x^2 \right)\\
 - \exp \left( -\pi \sum_{i \in \mathcal{K}} \lambda_i \left( \frac{P_i}{P_k} \right)^{2/\alpha} x^2\right)\Bigg]
\end{multline}
where the association probability to the REA of tier $k$ is $\pare = 1 - \sum_{i \neq k} p_{\mathrm{a}, i}(b) - p_{\mathrm{a}, k}(B\,|\,b_k = 1)$ and the association probability to the $k$-th tier is given by
\be
p_{\mathrm{a},k}(B) = \frac{\lambda_k}{\sum_{i \in \mathcal{K}}\lambda_i \left( \frac{P_i}{P_k}\frac{b_i}{b_k}\right)^{2/\alpha}} .
\ee
\end{lemma}
For the proof, we refer to \cite{JoSanXiaAnd_TWC12} and \cite{TanPenHonXue_WCL13}. Note that a user, which belongs without biasing to tier $i \neq k$, is located in  the REA $C_k^{\mathrm{(RE)}}$ of tier $k$ if the relationship $P_k x_k^{-\alpha} < P_i x_i^{-\alpha} < b_k P_k x_k^{-\alpha}$ holds. 
%\begin{proof}
%We consider an association policy based on the biased long-term received signal power, averaged over the fading distribution. Along the lines of Lemma 3 in \cite{JoSanXiaAnd_TWC12}, the CCDF of the distance distribution is given by $\PP[x_k > x \, | \, x_k \in C_k^{\mathrm{(RE)}}] = \PP[x_k > x , x_k \in C_k^{\mathrm{(RE)}}]/p_{\mathrm{a},C_k^{\mathrm{(RE)}}}$ with $C_k^{\mathrm{(RE)}}$ the REA of tier $k$ and $x_k$ the distance to the closest AP of tier $k$. A user, which belongs without biasing to tier $i \neq k$, is located in  $C_k^{\mathrm{(RE)}}$ if the following relationship holds
%\be \nonumber
%P_k x_k^{-\alpha} < P_i x_i^{-\alpha} < b_k P_k x_k^{-\alpha} .
%\ee
%The rest of the proof is similar to \cite{JoSanXiaAnd_TWC12}.
%\end{proof}
In order to calculate the benefit of canceling the strongest interferer for the users located in the REA, we provide the following lemma. 
\begin{lemma} \label{lem.PsRE}
After canceling the strongest AP, the success probability of the users located in $C_k^{\mathrm{(RE)}}$  is given by 
\begin{multline} \label{eq.PsicRE1}
\Psic(\etat,1 \,|\, x_k \in \borRE) = \frac{1}{\pare} \\
\times \Bigg[\frac{1}{\sum_{i \in \mathcal{K}} \left(\frac{\lambda_i}{\lambda_k} \right) \left(\frac{P_i}{P_k} \right)^{2/\alpha} \left(\etat^{2/\alpha} C((1/\etat)^{2/\alpha}, \alpha)+\left(\frac{b_i}{b_k} \right)^{2/\alpha}\right) } \\
 - \frac{1}{\sum_{i \in \mathcal{K}}  \left(\frac{\lambda_i}{\lambda_k} \right) \left(\frac{P_i}{P_k} \right)^{2/\alpha} \left(\etat^{2/\alpha} C((1/\etat)^{2/\alpha}, \alpha)+1\right)} \Bigg] .
\end{multline}
\end{lemma}
\begin{proof}
See Appendix \ref{app.lemPsRE}
\end{proof}
From Lemma \ref{lem.PsRE}, the bias factors of the different tiers can be determined to guarantee a given performance for the mobile users belonging to the REA. 

%%%%%%%%%%%%%%%%%%%%%%%%%%%%%%%%%%%%%%
\section{Numerical Results} \label{sec.NumRes}
In this section, we present some numerical results that illustrate that SIC, although not very effective in networks applying the association policy based on the maximum average received signal power, can have distinct advantages in scenarios with other association policies. Figure \ref{fig.rate_minLoad_SIC} depicts the coverage probability and compares the max-SIR association policy with the minimum load association policy. From the numerical results, we see that the coverage probability decreases significantly for the minimum load policy when no SIC is applied. This means that the loss in SIR cannot be compensated by the lower load of the AP. However, when SIC is applied ($n = 1$), the coverage probability is comparable with that of the max-SIR policy. From this figure, we conclude that when SIC is applied, users can be offloaded to nearby APs without loss of capacity, which paves the way for more advanced load balancing techniques.
\begin{figure}[t!]
\centering
\psfrag{lu = 5e-5 - closest BS}[][][0.5]{$\mu_j = 5e^{-5}$ - closest BS}
\psfrag{lu = 5e-5 - min load BS}[][][0.5]{$\,\mu_j = 5e^{-5}$ - min load BS}
\psfrag{lu = 5e-5 - min load BS SIC}[][][0.5]{$\,\,\mu_j = 5e^{-5}$ - min load BS SIC}
\psfrag{lu = 2e-4 - closest BS}[][][0.5]{$\mu_j = 5e^{-5}$ - closest BS}
\psfrag{lu = 2e-4 - min load BS}[][][0.5]{$\,\mu_j = 5e^{-5}$ - min load BS}
\psfrag{lu = 2e-4 - min load BS SIC}[][][0.5]{$\,\,\mu_j = 5e^{-5}$ - min load BS SIC}
\includegraphics[width = 1 \linewidth] {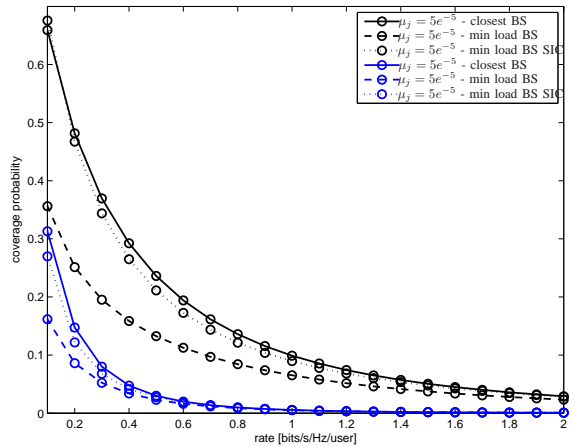}
\caption{The coverage probability is depicted for the max-SIR association policy (solid lines), the minimum load association policy (dashed lines), and the minimum load policy with SIC (dotted lines) for $\lambda = 10^{-5}$ and $\alpha = 4$.} 
\label{fig.rate_minLoad_SIC}
\end{figure}

Figure \ref{fig.Ps_maxSIR} shows the success probability for a typical user that connects to the AP that provides the highest instantaneous SIR, with and without SIC. We considered a two-tier network with densities $\lambda_1 = 10^{-5} \text{m}^{-2}$ and $\lambda_2 = 10^{-4} \text{m}^{-2}$, respectively, while the user density is given by $\mu = 10^{-4} \text{m}^{-2}$. The ratio between transmission powers is given by $P_1/P_2 = 10$ and $Q_1/Q_2 = 10$. From the figure, we observe an increase of the success probability up to $20\%$ in the range of interest, i.e. for values of the threshold above $0\, \text{dB}$. However, this scheme requires additional signaling between the APs to connect the user to the AP which provides instantaneously the highest SIR, which limits the practical feasibility.
\begin{figure}[t!]
\centering
\psfrag{Ps with SIC}[][][0.75]{$\PP_{\mathrm{s, SIC}}$}
\psfrag{Threshold (dB)}[][][0.75]{$\etat$ [dB]}
\psfrag{SIC - k = 0}[][][0.5]{SIC - $k = 0$}
\psfrag{SIC - k = 1}[][][0.5]{SIC - $k = 1$}
\psfrag{SIC - k = 2}[][][0.5]{SIC - $k = 2$}
\psfrag{SIC - k = 3}[][][0.5]{SIC - $k = 3$}
\includegraphics[width = 1 \linewidth] {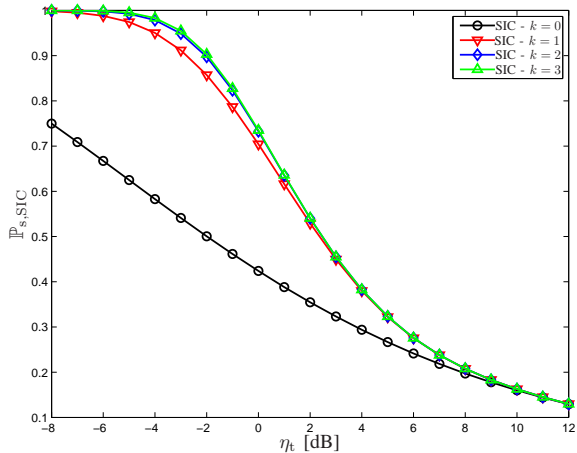}
\caption{The gain in success probability that can be achieved by successively canceling interferers.}
\label{fig.Ps_maxSIR}
\end{figure}

Figure \ref{fig.Ps_SIC_RE} depicts the success probability of a typical user in the REA in a two-tier network with densities $\lambda_1 = 10^{-5} \text{m}^{-2}$ and $\lambda_2 = 10^{-4} \text{m}^{-2}$ for different values of the range expansion factor $b$. The figure illustrates how the success probability decreases as the REA gets larger with increasing values of $b$. Moreover, from the numerical results we observe that the increase of success probability due to SIC is substantial. This scenario is a realistic example where SIC can provide high performance gain. 
\begin{figure}[t!]
\centering
\psfrag{coverage probability}[][][0.75]{success probability}
\psfrag{threshold (dB)}[][][0.75]{$\etat$ [dB]}
\psfrag{b = 2}[][][0.6]{$b=2$}
\psfrag{b = 5}[][][0.6]{$b=5$}
\psfrag{b = 10}[][][0.6]{$b=10$}
\includegraphics[width = 1 \linewidth] {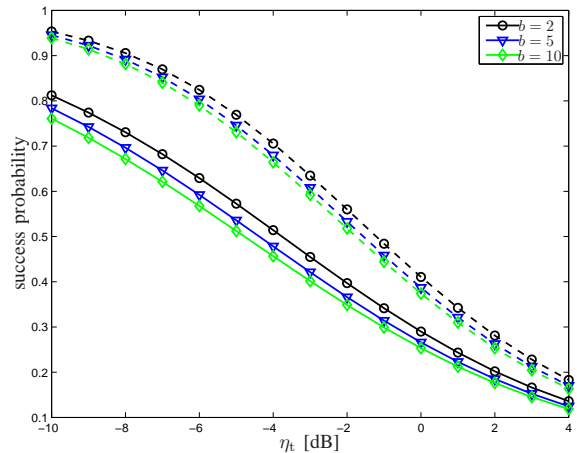}
\caption{Success probability for users belonging to the range expansion region. Solid lines represent the success probability without interference cancellation, whereas the dashes lines represent the success probability when the closest AP is canceled. The power ratio $P_1/P_2 = 10$.}
\label{fig.Ps_SIC_RE}
\end{figure}

%%%%%%%%%%%%%%%%%%%%%%%%%%%%%%%%%%%%%%
\section{Conclusions} \label{sec.Con}
In this work, we developed a probabilistic framework for the performance analysis of multi-tier heterogeneous networks with SIC capabilities. The framework accounts for the consecutive steps of the IC scheme, network topology, propagation effects, and the association policy. For users connected to the AP that provides the maximum average received signal power, performance benefits diminish quickly with the number of cancellations $n$ and are modest for realistic values of the SIR. Yet, we presented several deployment scenarios for future multi-tier networks where SIC yields distinct performance gains, such as for the minimum load association policy, the maximum instantaneous SIR policy, and range expansion. The proposed framework can be applied to define the achievable performance of multi-tier networks with SIC capabilities. Future work to extend the current framework will include imperfect interference cancellation and will model the decreasing cost of APs of higher tiers considering for each tier a different maximum number of interferer cancellations. 

%%%%%%%%%%%%%%%%%%%%%%%%%%%%%%%%%%%%%%
\section{Acknowledgment}
This work was partly supported by the SRG ISTD 2012037, CAS Fellowship for Young International Scientists Grant 2011Y2GA02,  SUTD-MIT International Design Centre under Grant IDSF1200106OH, the A$^\ast$STAR Research Attachment Programme, and the European Commission Marie Curie International Outgoing Fellowship under Grant 2010-272923.

%%%%%%%%%%%%%%%%%%%%%%%%%%%%%%%%%%%%%%
\appendices
%%%%%%%%%%%%%%%%%%%%%%%%%%%%%%%%%%%%%%
\section{Proof of Lemma \ref{lem.Par}} \label{app.lemPar}
We calculate the distribution of $Y = h X^{-\alpha}$. For the uniformly distributed distance variable $X$, the distribution is given by $F_X(x) = x^2/R^2$ such that $R$ is the maximum considered range where nodes still have a contribution to the aggregate interference. The distribution of $W = X^{-\alpha}$ is given by
\begin{align} \nonumber
F_W(w) &= \PP[X^{-\alpha}<w] \\ \nonumber
&= \PP[X \geq (1/w)^{1/\alpha}] \\ \nonumber
&= 1 - \frac{w^{-2/\alpha}}{R^2},
\end{align}
which corresponds to the CDF of a Pareto distribution. The distribution of the product of $W$ and the exponential r.v. $h$ is now given by
\begin{align}\nonumber
F_Y(y) &= \EE\left\lbrace \PP[W \leq y/h]\right\rbrace \\ \nonumber
& = \int_0^\infty \left(1 - \frac{(y/h)^{-2/\alpha}}{R^2} \right)\exp(-h)\mathrm{d}h \\ \nonumber
%& = 1 - \frac{y^{-2/\alpha}}{R^2} \int_0^\infty h^{2/\alpha} \exp(-h) dh\\
&= 1 - \Gamma\left( \frac{2}{\alpha}+1\right) \frac{y^{-2/\alpha}}{R^2},
\end{align}
which concludes the proof.

%%%%%%%%%%%%%%%%%%%%%%%%%%%%%%%%%%%%%%
\section{Proof of Lemma \ref{lem.PscanTSD}} \label{app.lemPscanTSD}
In \eqref{eq.TSD}, $\alpha_I$ is chosen equal to the characteristic exponent of the skewed stable distribution in the unbounded path loss model. The parameters of the TSD can be found using the method of the cumulants. From \eqref{eq.TSD}, the cumulants of the truncated stable distribution can be expressed as
\bea \nonumber \label{eq.kappa1}
\kappa_I(k) &=& \frac{1}{j^k}\frac{d^k}{d\omega^k} \ln \psi_{I_{\Omega_j^n}}(j\omega)\Big|_{\omega = 0 } \\
&=& (-1)^k \gamma'\Gamma(-\alpha_I) g^{\alpha_I-k} \Pi_{i=0}^{k-1}(\alpha_I -i) .  
\eea
Building on Campbell's theorem \cite{RabQueShiWin_JSAC2011}, the cumulants of the aggregate interference can be expressed as
\be \label{eq.kappa2}
\kappa(k) = Q^k\frac{2 \pi \muj}{k\alpha-2} d_{\mathrm{min}}^{2-k\alpha} \mu_{h^2}(k) .
\ee
Using \eqref{eq.kappa1} and \eqref{eq.kappa2}, the parameters $\gamma'$ and $g$ can be written as a function of the first two cumulants as follows
\bea \nonumber
\gamma' &=& \frac{-\kappa(1)}{\Gamma(-\alpha_I)\alpha_I\left( \frac{\kappa(1)(1-\alpha_I)}{\kappa(2)} \right)^{\alpha_I-1}} \\
g &=& \frac{\kappa(1)(1-\alpha_I)}{\kappa(2)} .
\eea
The success probability of successfully canceling the strongest interferer can be written as
\begin{align} \nonumber
\PP_{\s, \text{can}}(\etat, n | r)  & = \PP\left( \frac{X_{(n)}}{I_{\Omega_j^n}} \geq \etat \right) \\ 
\nonumber& = \mathcal{L}_{I_{\Omega_j^n}}(\etat r^\alpha) \\ 
& = \exp\left( \gamma' \Gamma(-\alpha_I)[(g+\etat r^\alpha)^{\alpha_I} - g^{\alpha_I} ] \right) .
\end{align}
For the special case where $\alpha = 4$, we get $\alpha_I = 2/\alpha = 1/2$, and $\Pscan(\etat, n | r)$ conditioned on the distance of the strongest node is given by
\begin{align} \nonumber
&\Pscan(\etat, n|r, \alpha = 4) \\ \nonumber
&= \exp\left( \gamma' \Gamma(-1/2)  \sqrt{g} \left[ \sqrt{1 + \frac{\etat}{g}r^4} -1 \right] \right) \\ \nonumber
 & = \exp\left( -2 \kappa(1) \sqrt{\frac{\kappa(1)}{2\kappa(2)}}  \sqrt{\frac{\kappa(1)}{2\kappa(2)}}  \left[\sqrt{1 + \frac{2\etat r^4 \kappa(2)}{\kappa(1)}} -1\right] \right) \\ \nonumber
 & = \exp\left(- \frac{\kappa(1)^2}{\kappa(2)} \left[\sqrt{1 + \frac{2\etat r^4 \kappa(2)}{\kappa(1)}} -1\right] \right) \\ 
 & \overset{(a)}{=} \exp(-3/2 \muj \pi r^2 [\sqrt{1 + 4 \etat/3}-1])\, ,
\end{align}
where $(a)$ follows from the fact that $d_{\text{min}}$ corresponds to the distance $r$ of the $n$-th interferer, and we assume Rayleigh fading such that $\EE[h^k] = k!/\lambda^k$ with $\lambda = 1$ the intensity of the exponential distribution. The unconditional success probability can now be written as
\begin{align} \nonumber
&\Pscan(\etat, n|\alpha=4) \\ \nonumber
&= \int_0^\infty \exp(-3/2 \muj \pi r^2 [\sqrt{1 + 4 \etat/3}-1]) \\ \nonumber
& \hspace{13em} \times \exp(\muj \pi r^2) \frac{2(\muj \pi r^2)^n}{r \Gamma(n)} \mathrm{d}r \\ \nonumber
& = \int_0^\infty \exp(\muj \pi r^2[\sqrt{9/4 + 3\etat} - 3/2 + 1]) \frac{2 (\muj \pi)^n r^{2n-1}}{\Gamma(n)} \mathrm{d}r \\ \nonumber
& = \int_0^\infty \exp(\muj \pi r^2[\sqrt{9/4 + 3\etat} - 1/2]) \frac{(\muj \pi)^{n-1} r^{2n-2}}{\Gamma(n)}\mathrm{d}\muj \pi r^2 \\ \nonumber
& = \int_0^\infty \exp([\sqrt{9/4 + 3\etat} - 1/2]w) \frac{w^{n-1}}{\Gamma(n)} \mathrm{d}w \\
& = \frac{1}{(\sqrt{9/4+3\etat}-1/2)^n} \quad .
\end{align}

%%%%%%%%%%%%%%%%%%%%%%%%%%%%%%%%%%%%%%
\section{Proof of Lemma \ref{lem.PsRE}} \label{app.lemPsRE}
The success probability of a mobile node belonging to $C_k^{(\text{RE})}$ and connected to the $k$-th tier conditioned on the distance can be written as
\be \nonumber
\PP_\s(\etat \, | \, x_k \in \borRE, x_k) = \prod_{i \in \mathcal{K}} \mathcal{L}_{I_{\Phi_i}}\left(\frac{\etat x_k^\alpha}{P_k} \right),
\ee
where 
\begin{multline} \nonumber
 \mathcal{L}_{I_{\Phi_i}}\left(\frac{\etat x_k^\alpha}{P_k} \right) \\
 = \exp \left( -\pi \lambda_i \etat^{2/\alpha}\left(\frac{P_i}{P_k} \right)^{2/\alpha} C((b_i/\etat b_k)^{2/\alpha},\alpha)  x_k^2\right) \quad. 
\end{multline}
The integration interval of the integral in $C(b,\alpha)$ is determined noting that the location of the user in the REA $P_i x_i^{-\alpha} < b_k P_k x_k^{-\alpha}$ yields the interferer exclusion region $x_i > (P_i/b_kP_k)^{1/\alpha} x_k$. Applying the change of variables $(P_k/\etat P_i)^{2/\alpha}(x_i/x_k)^2 \to u$ and deconditioning  on $x_k$, we can write
\begin{align} \label{eq.Ps_RE} \nonumber
&\PP_\s(\etat\,|\,x_k \in \borRE) \\ \nonumber
& = \int_0^\infty \exp(-\pi \etat^{2/\alpha} \sum_{i\in \mathcal{K}} \lambda_i (P_i/P_k)^{2/\alpha} C((b_i/\etat b_k)^{2/\alpha},\alpha)x_k^2) \\ \nonumber
& \hspace{1 em} \times f_{X_k^{\mathrm{(RE)}}}(x_k) \mathrm{d}x_k \\ \nonumber
& = \frac{1}{\pare} \\ \nonumber
&\times \Bigg(\frac{1}{\sum_{i \in \mathcal{K}} \left(\frac{\lambda_i}{\lambda_k} \right) \left(\frac{P_i}{P_k} \right)^{2/\alpha} \left(\etat^{2/\alpha} C((b_i/\etat b_k)^{2/\alpha}, \alpha)+\left(\frac{b_i}{b_k} \right)^{2/\alpha}\right) } \\
 & - \frac{1}{\sum_{i \in \mathcal{K}}  \left(\frac{\lambda_i}{\lambda_k} \right) \left(\frac{P_i}{P_k} \right)^{2/\alpha} \left(\etat^{2/\alpha} C((b_i/\etat b_k)^{2/\alpha}, \alpha)+1\right)} \Bigg) \, .
\end{align}
Applying SIC to a user located in the REA, the highest unbiased received signal power of tier $i$ is canceled. As a result the interference cancellation radius relative to the $i$-th tier increases from $(P_i/b_kP_k)^{1/\alpha} x_k$ to $(P_i/P_k)^{1/\alpha} x_k$, and hence, the success probability after canceling the strongest AP can be written as \eqref{eq.PsicRE1}.

%%%%%%%%%%%%%%%%%%%%%%%%%%%%%%%%%%%%%%
%\begin{acronym}
%\acro{JD}{joint detection}
%\end{acronym}

%%%%%%%%%%%%%%%%%%%%%%%%%%%%%%%%%%%%%%
\bibliographystyle{IEEEtran}
%\bibliography{IEEEabrv,StringDefinitions,PCP}
\bibliography{./IEEEabrv,./StringDefinitions,./multiTierInterference}
%\bibliography{multiTierInterference}
\end{document}